\documentclass[sigconf,natbib=true,screen=true]{acmart}

\usepackage{multirow} %
\usepackage[inline]{enumitem}
\usepackage{mleftright}
\usepackage{amsthm} %
\newtheorem{theorem}{Theorem}

\AtBeginDocument{%
  }

\copyrightyear{2025}
\acmYear{2025}
\setcopyright{cc}
\setcctype{by}
\acmConference[RecSys '25]{Proceedings of the Nineteenth ACM Conference on Recommender Systems}{September 22--26, 2025}{Prague, Czech Republic}
\acmBooktitle{Proceedings of the Nineteenth ACM Conference on Recommender Systems (RecSys '25), September 22--26, 2025, Prague, Czech Republic}\acmDOI{10.1145/3705328.3748090}
\acmISBN{979-8-4007-1364-4/2025/09}

\begin{document}

\allowdisplaybreaks

\newcommand{\R}{\mathbb{R}}

\newcommand{\argmin}{\mathop{\mathrm{argmin}}}

\newcommand{\parahead}[1]{\noindent \textbf{#1.}}

\bibliographystyle{ACM-Reference-Format}

\title{A Non-Parametric Choice Model That Learns How~Users~Choose~Between~Recommended~Options}

\author{Thorsten Krause}
\orcid{0000-0002-3830-7708}
\email{thorsten.krause@ru.nl}
\affiliation{%
  \institution{Radboud University}
  \city{Nijmegen}
  \country{The Netherlands}
}
\author{Harrie Oosterhuis}
\orcid{0000-0002-0458-9233}
\email{harrie.oosterhuis@ru.nl}
\affiliation{%
   \institution{Radboud University}
  \city{Nijmegen}
  \country{The Netherlands}
}

\begin{abstract}
Choice models predict which items users choose from presented options.
In recommendation settings, they can infer user preferences while countering exposure bias.
In contrast with traditional univariate recommendation models, choice models consider which competitors appeared with the chosen item.
This ability allows them to distinguish whether a user chose an item due to \emph{preference}, i.e., they liked it; or \emph{competition}, i.e., it was the best available option.
Each choice model assumes specific user behavior, e.g., the multinomial logit model.
However, it is currently unclear how accurately these assumptions capture actual user behavior, 
how wrong assumptions impact inference,
and whether better models exist.

In this work, we propose the \emph{learned choice model for recommendation (LCM4Rec)}, a non-para\-metric method for estimating the choice model.
By applying kernel density estimation, LCM4Rec infers the most likely error distribution that describes the effect of inter-item cannibalization and thereby characterizes the users' choice model.
Thus, it simultaneously infers \emph{what} users prefer and \emph{how} they make choices.
Our experimental results indicate that our method (i) can accurately recover the choice model underlying a dataset; (ii) provides robust user preference inference, in contrast with existing choice models that are only effective when their assumptions match user behavior; and (iii) is more resistant against exposure bias than existing choice models. 
Thereby, we show that learning choice models, instead of assuming them, can produce more robust predictions. 
We believe this work provides an important step towards better understanding users' choice behavior.
\end{abstract}



\maketitle

\section{Introduction}\label{Sect. Introduction}

Traditional implicit collaborative filtering methods assume that all items are chosen independently of each other \citep{rendle2012bpr, petrov2023gsasrec}. 
Accordingly, they estimate an item's relevance score purely from information about which users interacted with it~\citep{hu2008collaborative}. 
In reality, users must often decide for one of multiple items that the recommender system exposes together, for example, when ordering lunch \citep{li2024recommender}, booking a flight \citep{kan2024personalized}, or choosing a movie~\citep{steck2021deep}. 
Consequently, the choices for one item and against others are dependent.

One prominent class of models than accounts for dependencies between choices are so-called choice models.
Choice models predict the choice probabilities according to the preference towards each option, while also considering the competitive effect between options.
They model this competitive effect through a probabilistic noise distribution, which is their key-characteristic. For instance, the multinomial logit model (MNL) \citep{train2009discrete}, which is equivalent to the softmax, is characterized by its Gumbel noise distribution.
Many sequential models, such as BERT4Rec \citep{sun2019bert4rec}, employ the softmax loss to predict the next user interaction. 
Furthermore, \citet{krause2024mitigating} showed that choice models are significantly more robust against exposure bias, particularly against exposure bias from competition among items.
By providing choice probabilities, choice models can also contribute to related tasks such as balancing trade-offs in responsible recommendation~\citep{g2024overcoming} or, potentially, informing \emph{whether} to make a recommendation and \emph{how many} items to recommend.

However, existing approaches are only effective when the assumed choice model matches the true user behavior.
Employing the wrong choice model can bias choice probabilities \citep{train2009discrete}. 
This poses a significant limitation, since the most accurate choice model for a given use case is generally unknown.
Identifying it through benchmarking is practically infeasible since infinitely many possible choice models exist and some choice models, such as the popular probit, lack a closed-form analytical solution \citep{buscher2024weighting}.

In this work, we directly address this limitation by proposing the \emph{learned choice model for recommendation (LCM4Rec)}, a novel user-modeling method that simultaneously learns the most likely choice model and user preferences.
Instead of assuming the probabilistic noise distribution underlying the true users' choice model, we utilize a non-parametric kernel density estimator to approximate it.
Specifically, LCM4Rec maximizes the log-likelihood of observed interaction data, and thereby, infers both the most likely utility distribution (i.e., the user preferences) and noise distributions (i.e., the choice model) underlying their interactions.
Because kernel density estimators can be used to approximate any arbitrary probability distribution, LCM4Rec is applicable to any possible choice model.\footnote{Any choice model for which the error terms are IID.}
Moreover, because the choice model is inferred from interaction data, we do not rely on a priori assumptions about what choice model users follow,\footnote{To be precise, we only make the assumption that choice model error terms are IID.} in contrast with previous work~\citep{krause2024mitigating}.
Additionally, we investigate the negative implications from assuming an inaccurate choice model with respect to ranking, choice probability estimation and exposure bias robustness.

In summary, LCM4Rec resolves the need for benchmarking or assuming choice models and is applicable to any choice model.$^2$ %
It simultaneously learns user preferences towards items as well as the competitive effects between presented items.
Our experimental results on synthetic datasets indicate LCM4Rec is robust to what choice model underlies interaction data.
Whilst existing methods can only accurately infer preferences when their assumed choice model matches the true model, LCM4Rec provides reliable and robust performance regardless of the underlying model.
Furthermore, our results show that LCM4Rec is better at mitigating exposure bias.
We also found that employing inaccurate choice models skews choice probabilities and increases exposure bias vulnerability. 

To the best of our knowledge, we introduce the first method that both learns users preferences as how users choose between recommended options.
Besides the improved effectiveness and robustness of LCM4Rec, we believe our work contributes to the better modeling and understanding of user behavior whilst avoiding assumptions.

\section{Related work and Background}\label{Sect. Related Work}

\subsection{Discrete choice models for recommendation}
Discrete choice models estimate with which probability $P_{ij|C}$ a user $i\in I$ would choose an item $j\in J$ from a discrete choice set $C \in \mathcal{C}$ by utilizing the transitive information on which items are preferred over another. 
These models assume that each user would always chose the choice alternative $j \in C$ from which they expect the highest utility $U_{ij}$. 
Econometricians use discrete choice models for determining how a variable affects preferences \citep{haghani2021landscape} and operations researchers use them for solving optimization problems based on observed choice data \citep{corner1991decision}.
Discrete choice models have also been used for recommendation, most prominently under the term \emph{collaborative competitive filtering} \citep{yang2011collaborative}.
Their key difference from classical implicit feedback collaborative filtering models is that they require information about displayed choice alternatives.
In recent years, more such data has become available, for example in the \emph{finn.no} dataset~\cite{eide2021finn} or the 2024 RecSys Challenge dataset\footnote{https://recsys.eb.dk/dataset/}.

To account for the fact that humans do not to strictly chose transitively~\citep{busemeyer2014psychological}, discrete choice models are probabilistic. 
They model the total utility $U_{ij}$ as the sum of a deterministic utility $V_{ij}$ and a random error term $\epsilon_{ij}$ s.t.\ $U_{ij} = V_{ij} + \epsilon_{ij}$. 
In the recommendation context, the deterministic utility $V_{ij}$ corresponds to the relevance scores. 
The random error is interpreted as the share of utility that analysts cannot observe~\citep{train2009discrete}.
Thereby, it enables discrete choice models to capture that identical circumstances can result in different choices, or that succeeding choices can appear contradicting.

Our study focuses on the shape of the distribution that underlies $\epsilon_{ij}$.
The most commonly assumed error distribution is the IID Gumbel distribution of the multinomial logit model which results in the following softmax-like choice probabilities ~\citep{train2009discrete}:
\begin{equation}
    P_{ij \mid C} = 
    \frac{
        \exp\left(V_{ij}\right)
    }
    {
        \sum_{j'\in C}\exp\left(V_{ij'}\right)
    }.
\end{equation}
The closed analytical form and easy interpretability of its choice probabilities have contributed to the popularity of the multinomial logit model~\citep{train2009discrete}.
Its key property is its assumption of the irrelevance of independent alternatives (IIA), which states that the ratio of choice alternatives between two items does not depend on any third item's utility ~\citep{ben1999discrete}.
However, this key property is also a source of criticism ~\citep{ben1985discrete}.
Empirical studies indicate that the IIA assumption is inaccurate in many settings~\citep{benson2016relevance}.

\subsection{Non-parametric discrete choice models}\label{Subsect. Non-parametric likelihood estimation}
Instead of having the analyst guess the correct choice model, non-parametric estimators learn the error distribution.
This approach builds on the idea that the choice probabilities can be re-written as:
\begin{equation} \label{Eq. Choice probabilities no integral}
\begin{split}
    P_{ij\mid C} &= P(\forall j' \in C,\; j' \neq j \rightarrow U_{ij} > U_{ij'}) \\
     &= P(\forall j' \in C,\; j' \neq j \rightarrow V_{ij} + \epsilon_{ij} - V_{ij'} > \epsilon_{ij'}).
\end{split}
\end{equation}
Under the assumption that all error terms $\epsilon_{ij}$ are IID distributed to some random variable $\epsilon$ with cdf $F_\epsilon$ and pdf $\rho_\epsilon$, \citet{manski1975maximum} points out that the choice probabilities can be further reformulated as:
\begin{equation} \label{Eq. Likelihood}
    P_{ij\mid C} = \int \rho_{\epsilon} (e_{ij}) \prod_{j' \in C, j' \neq j} F_\epsilon(V_{ij} + e_{ij} - V_{ij'}) de_{ij},
\end{equation}
and proposes a \emph{maximum score estimator} that estimates a ranking function and that enables computing the effect of observed variables on the utility $V_{ij}$ consistently in the binary case with two choice options.
\citet{cosslett1983distribution} criticizes the maximum score estimator for being inapplicable to inferring choice probabilities and instead provides an estimator for the cdf $F_\epsilon$ in the binary case.
As outlined for the multivariate case by~\citet{briesch2010nonparametric}, the core idea is that given the deterministic utilities $V_{ij}$, the choice distribution exclusively depends on $F_\epsilon$ and $\rho_\epsilon$.
Thus, for $N$ observations ~$\{(i_n, j_n , C_n)\}_{n \leq N} \subseteq I \times J \times \mathcal{C}$, one can estimate $\epsilon$ by defining a broad family of distributions $\mathcal{F}$, and subsequently select the member which minimizes the negative log-likelihood:
\begin{equation}\label{Eq. Negative log-likelihood}
    \hat{F} \coloneqq \argmin_{F \in \mathcal{F}} - \sum_{n = 1}^N \log( P_{i_n j_n \mid C_n} ) %
    .
\end{equation} 
However, \citet{cosslett1983distribution} uses sorting and linear interpolation for estimating $F$, which is intractable for gradient descent. They also only consider the binary case. 
To the best of our knowledge, no other existing approach is directly applicable to the recommender systems setting either.
\citet{matzkin1992nonparametric} drops any assumptions on the shape of $V$, but derives a discrete optimization problem that returns a non-differentiable estimator.
\citet{horowitz1992smoothed} smoothes the maximimum score estimator but still cannot infer choice probabilities.
For the binary case, \citet{klein1993efficient} estimate the choice probabilities $P_{ij|C}$ directly. However, generalizing their method to the multivariate case requires density estimation with as many dimensions as choice alternatives, which is unfeasible for our setting where choice sets can contain 20 items or more \citep{eide2021finn}.
More recent works focus increasingly on challenges that are beyond current concerns in recommender systems research. For example, \citet{hirano2002semiparametric} develops a solution for random effect autoregressive models with non-parametric idiosyncratic shocks. 
Others use techniques that we deem incompatible with gradient descent-based optimization such as discrete estimators \citep{briesch2010nonparametric, ichimura1998maximum}, or linear programming \citep{chiong2017counterfactual, tebaldi2023nonparametric}.

\begin{table}[t]
  \caption{
  Cannibalization patterns according to two different choice models. For items $j_1$, $j_2$, $j_3$ with utilities $V_{ij_1} = 3$, $V_{ij_2} = 1$, $V_{ij_3} = 2$,
  the models agree on the choice probabilities between $j_1$ and $j_2$, but disagree when the option $j_3$ is also available.
  }
  \label{Tab. cannibalization example}
  \begin{tabular}{cclll}
    \toprule
    \multirow{ 2}{*}{Model} & \multirow{2}{*}{Choice set} & \multicolumn{3}{c}{Choice probability} \\
     &  & $j_1$ & $j_2$ & $j_3$ \\
    \midrule
    \multirow{ 2}{*}{Multinomial logit} & $\{j_1, j_2\}$ & 88.1\% & 11.9\% & — \\
     & $\{j_1, j_2, j_3\}$ & \textbf{66.5\%} & \textbf{9.0\%} & 24.5\% \\
     \midrule
    \multirow{ 2}{*}{Exponomial} & $\{j_1, j_2\}$ & 88.1\% & 11.9\% & —\\
     & $\{j_1, j_2, j_3\}$ & \textbf{70.2\%} & \textbf{5.3\%} & 24.5\% \\
  \bottomrule
\end{tabular}
\end{table}

\subsection{Competitive effects and cannibalization} \label{Subsect. Cannibalization}
The choice model also describes the competitive effect between items, i.e., how the availability and utility of other item options affects choice probabilities \citep{alptekinouglu2016exponomial,chiong2017counterfactual}.
Specifically, the concept of cannibalization refers to how the choice probabilities of items decrease due to a substitute's presence \citep{lipovetsky2014finding, alptekinouglu2016exponomial}.
Each choice model predicts a different effect in such cases.
For example, under the multinomial logit model all choice probabilities decrease by the same proportion, whereas under the exponomial choice model, low-utility alternatives lose a larger share of their choice probabilities than high-utility alternatives.
Table~\ref{Tab. cannibalization example} displays the differences between these models in an example scenario based on the one in \citet{alptekinouglu2016exponomial}.
Whilst both models agree in the ordering of the most probable choices, their predicted choice probabilities differ.
Consequently, in order to accurately infer user preferences, the choice model should match the true user behavior.
\citet{alptekinouglu2016exponomial} argue that the exponomial model better represents real-world cannibalization patterns than the multinomial logit model, while referring to the empirical observations by \citet{blattberg1989price}.

\subsection{Exposure bias in recommender systems}

Exposure bias occurs whenever a model's previous recommendation affect user behavior and subsequently the logged feedback data. 
Whilst different formal definitions of exposure bias exist, we follow previous work \citep{wan2022cross, krause2024mitigating} that considers a model \emph{exposure bias resistant} if the exposure distribution does not affect the model's asymptotic rankings.
This definition does not imply convergence towards the true ranking because other biases could still occur.

Mitigating exposure bias should be a priority for all stakeholders. 
It manipulates future model generations to recommend similar items as in the past, resulting in user-model feedback loops \cite{chen2023bias, jiang2019degenerate}. 
Through exposure bias, the model over-proportionally prefers popular items \cite{ferraro2019artist} and can fail to cater to the special interests of user subgroups \cite{abdollahpouri2020connection}. 
In the long-term, its outputs to different user groups can homogenize \cite{mansoury2020feedback} and in multi-stakeholder environments, the unfair exposure can translate into intra-provider unfairness \cite{abdollahpouri2020multi, dash2021umpire}.

\citet{krause2024mitigating} demonstrated that, in addition to \emph{overexposure}, which occurs when some items are exposed more often than others, exposure bias can result from \emph{competition}, which occurs when some items tend to be exposed with popular or unpopular alternatives.
They also demonstrated that discrete choice models can effectively reduce both forms of exposure bias. 
Their formal proof for unbiasedness, however, assumed that the model matches a ground-truth choice model. While using an inaccurate model could impair unbiasedness, \citeauthor{krause2024mitigating} did not further investigate implications for exposure bias.

\section{Method: The Learned Choice Model for Recommendation (LCM4Rec)}\label{Sect. Non-parametric Estimators}

Our goal is to propose a choice model approach for modeling user choices in recommender settings while avoiding assumptions about user behavior.
The non-parametric choice models discussed in Section~\ref{Subsect. Non-parametric likelihood estimation} prevent the need to assume users' true choice model, and thus appear to provide a great alternative to their parametric counterparts.
However, each is incompatible with state of the art recommender systems due to at least one of the following reasons:
\begin{enumerate*}[label=(\roman*)]
\item they utilize non-differentiable estimators and optimization methods;
\item they are only applicable to the binary choice cases where users are only presented with two options; or %
\item they are computationally complex for application on large choice sets.
\end{enumerate*}
As a result, the existing non-parametric methods do not provide any approach that is relevant to current recommender systems.

In this section, we address these shortcomings directly and propose our novel non-parametric choice model: the \emph{learned choice model for recommendation (LCM4Rec)}.
Its core feature is a smooth kernel density estimator that approximates the cdf $F_\epsilon$ and the pdf $\rho_\epsilon$ of the error distribution (Equation \ref{Eq. Likelihood}).
Importantly, the estimator can approximate any arbitrary error distribution and provide a differentiable likelihood function for them; thereby, LCM4Rec can approximate any possible choice model with IID error terms.
Moreover, LCM4Rec can be optimized to predict observed interaction data, which enables learning what choice model best explains user behavior, instead of assuming this a priori.
To the best of our knowledge, this makes LCM4Rec the first non-parametric choice model that is compatible with state-of-the-art recommender systems.

\subsection{Learned choice model architecture} \label{Subsect. Learned choice model architecture}
In order to approximate any arbitrary error distribution underlying a choice model, we use kernel density estimation to construct a family of functions $\mathcal{F}$ that contains approximately close members to any possible cdf $F_\epsilon$.
Subsequently, we search for a member of $\mathcal{F}$ that minimizes the negative log-likelihood of observed interaction data (equation $\ref{Eq. Negative log-likelihood}$).
We start our description of this procedure by defining and discussing our choice for $\mathcal{F}$.
Consider the family $\mathcal{F}$ of sigmoid kernel-based functions with the following shape:
\begin{equation}
     \hat{F}_K(x) \coloneqq \sum_{k=1}^K w_k \sigma\mleft( \frac{x - x_k}{h_k} \mright) \in C^\infty \mleft( \mathbb{R} \mright),
     \label{eq:cdfmodel}
\end{equation}
where $\sigma$ is the sigmoid function, $K \in \mathbb{N}$, $x_k \in \mathbb{R}$, $w_k \in \mathbb{R}_{>0}$, $h_k \in \mathbb{R}_{>0}$ and $\sum_{k=1}^K w_k = 1$.
In other words, each $\hat{F}_K$ is a weighted average of sigmoid functions scaled and translated by the $h_k$ and $x_k$ parameters.
Importantly, the derivative of $\hat{F}_K$ is numerically stable:
\begin{equation}
    \hat{\rho}_K(x) = \sum_{k=1}^K \frac{w_k}{h_k} \sigma\mleft( \frac{x - x_k}{h_k} \mright) \mleft( 1 - \sigma\mleft( \frac{x - x_k}{h_k} \mright) \mright) \in C^\infty(\mathbb{R}).
\end{equation}
Any function $\hat{F}_K \in \mathcal{F}$ is strictly monotonically increasing and bounded in the interval $[0, 1]$. 
Moreover, $\hat{F}_K$ and $\hat{\rho}_K(x)$ are continuously differentiable with respect to their parameters and to the utility estimates $V_{ij}$ so that we can search for optimal representatives with regard to the empirical log-likelihood function  (Equation~\ref{Eq. Negative log-likelihood}) via gradient descent.

While directly learning the parameters $x_k$ is possible, not restricting them leads to local minima w.r.t. the loss function. 
For this reason, we set $x_k$ to be uniformly spread out over a closed interval: 
\begin{equation}
    x_1 = -l,\ x_2 = l \mleft(\frac{2}{K-1} - 1\mright),\ \dots,\ x_K = l,
    \label{eq:uniformspread}
\end{equation}
for a learnable scale parameter $l>0$. 
This makes optimization much easier without restricting the number of learnable functions when $N$ is large enough.
To enforce $w_k$, $h_k$, and $\omega$ to be positive, we introduce $\alpha_k \in \mathbb{R}$, $\beta_k \in \mathbb{R}$, and $\lambda \in \mathbb{R}$, and define:
\begin{equation}
    w_k \coloneqq \mathrm{softmax}\mleft( \alpha_k \mright), \;\,
    h_k \coloneqq \frac{l}{K} \mathrm{softplus}\mleft( \beta_k \mright), \;\,
    l \coloneqq \mathrm{softplus}\mleft( \lambda \mright).
\end{equation}
For numerical stability, we bound $-0.1 < \beta_i < 5$ and $ 0.1 < \lambda < 10$.
These bounding constraints ensure that the estimator does not collapse during the first iterations when the estimated utilities are very small.
With our model of the error distribution, we still require a model of the utilities $V_{ij}$; our definition of these applies commonly-used learned-embedding interactions: 
\begin{equation}
    V_{ij} \coloneqq u_i \cdot v_j + c_j, \label{eq:utilitymodel}
\end{equation}
where $u_i \in \mathbb{R}^m$ and $v_j \in \mathbb{R}^m$ with $m \in \mathbb{N}$ are user- and item-embeddings, and $c_j$ is an item-specific constant as is common in discrete choice modeling.

This completes the model architecture underlying LCM4Rec:
the error distribution of $\epsilon$ is modelled using a kernel density estimator that applies a weighted average of sigmoid functions (Equation~\ref{eq:cdfmodel}) where an error term $\epsilon_{ij}$ is IID sampled per user-item combination;
the user utilities are modelled by the dot-product between two learned embeddings and an item-specific constant (Equation \ref{eq:utilitymodel}); and the final utility is simply the sum of these: $U_{ij} = V_{ij} + \epsilon_{ij}$.
Recall that choice models assume the item with the highest utility is chosen (see Section~\ref{Sect. Related Work}).
Whilst our LCM4Rec is extremely expressive and can capture any arbitrary error distribution, as we show below, we have restricted it by assuming the error terms $\epsilon_{ij}$ are IID independent.
We note that this restriction rules out some choice models such as the nested logit model~\citep{train2009discrete}; we leave the extension of our model to the more general non-IID setting as future work.

Finally, we show that LCM4Rec can approximate any choice model with IID error terms.
Specifically, we show that for any cdf $F$, there always exists a member $\hat{F} \in \mathcal{F}$ that is arbitrarily close to $F$.
\begin{theorem}
Let $F$ be a continuous, strictly monotone cdf with finite support on a closed interval $S$. Let $\mathcal{F}$ be defined as above. 
Then, for all $\epsilon>0$, there exists a member $\hat{F} \in \mathcal{F}$ such that
\begin{equation}
    |\hat{F}(x) - F(x)| < \epsilon, \;\forall x \in S.
\end{equation}
\end{theorem}
\begin{proof}
Because $F$ operates on a closed interval, it is uniformly continuous. We can therefore find $\delta > 0$ such that 
\begin{equation}
    |F(x)-F(y)|~<~\frac{\epsilon}{4}~\;~\forall x,y \in S: \left| x - y \right| < \delta.
\end{equation}
Based on equation \ref{eq:cdfmodel}, set the number of kernels $K$ and the width $\lambda$ large enough so that any $x \in S$ is within close range of two design points, i.e., $\exists k_1, k_2 \leq K: x_{k_1} < x < x_{k_2}$ with~$0 < x_{k_2} - x_{k_1} < \delta$.
Set the weights $w_k$ so that~$\sum_l^k w_l = F(x_k), \; \forall k \leq K$. Then,
\begin{equation}
    \hat{F}_K(x_{k}) = F(x_{k-1}) + \underbrace{\frac{F(x_{k}) - F(x_{k-1})}{2}}_{>0} + \xi \leq F(x_{k})~\;~\forall 2\leq  k \leq K.
\end{equation}
where $\xi$ is a small error term that results from the sigmoids that are placed around the other design points $x_{k'}$ with $k'\neq k$. We can ignore $\xi$ as $\lim_{h \to 0} \xi = 0$.
Hence,
\begin{equation}
    F(x_{k-1}) \leq \hat{F}_K(x_{k}) \leq F(x_{k}) \quad \forall\, 2 \leq k \leq K \text{ and } h \ll 1.
    \label{eq:approximation_bounds}
\end{equation}
For $x \in S$ and its nearest point $k^* \coloneqq \underset{k \leq K}{\mathrm{argmax}}\mleft(x_k \leq x\mright)$ we get:
\begin{align}
&|\hat{F}_K(x) - F(x)| \\
&\underset{\triangle \text{-ineq.}}{\leq} \!\!\! |\hat{F}_K(x) - \hat{F}_K(x_{k^*})| + |\hat{F}_K(x_{k^*}) - F(x_{k^*})| + \mleft|F(x_{k^*}) - F(x)\mright|  \nonumber \\
&\underset{\text{(\ref{eq:approximation_bounds})}}{\leq} \! \mleft|F(x_{k^*+1}) \! - \! F(x_{k^*-1})\mright| \! + \! \mleft|F(x_{k^*+1}) \! - \! F(x_{k^*}\!)\mright| \! + \! \mleft|F(x_{k^*}\!) \! - \! F(x_{k^*\!-1})\mright| \! \nonumber \\
&\underset{\triangle \text{-ineq.}}{\leq} \!\!\! 2 \mleft|F(x_{k^*}) - F(x_{k^*-1})\mright| + 2 \mleft|F(x_{k^*+1}) - F(x_{k^*})\mright|
    < \epsilon. \nonumber
\end{align}
\end{proof}
\noindent Therefore, for any continuous, strictly monotone cdf $F$ with bounded support, the function family $\mathcal{F}$ contains an arbitrarily good approximator $\hat{F}$. 

\subsection{Optimizing for the most likely choice model}
Having constructed the family of functions $\mathcal{F}$,
we now introduce a method for finding a good candidate $\hat{F}_K \in \mathcal{F}$ for approximating $F_\epsilon$, based on $N$ observed choices $\{(i_n, j_n , C_n)\}_{n \leq N} \subseteq I \times J \times \mathcal{C}$ of users, chosen items, and choice sets.
Importantly, any $\hat{F}_K$ is differentiable with respect to its parameters $\alpha_k$, $\beta_k$, $\lambda$, and the user- and item-parameters $u_i$ $v_j$, and $c_j$.
Accordingly, we propose to learn all model parameters by optimizing the negative log-likelihood (NLL) function in Equation~\ref{Eq. Negative log-likelihood} via gradient descent.

However, the main challenge to this approach is that we lack a closed form for the integrals in Equation~\ref{Eq. Negative log-likelihood}, and thus, we cannot evaluate them directly.
As a solution, we propose to approximate its gradient via Monte Carlo integration.
Unfortunately, there is no closed form solution for the inverse cdf of $F_K$, which means we cannot sample from it directly.
Instead, we use the fact that $F_K$ is a weighted average of sigmoid cdfs from which samples can be drawn directly.
Specifically, for each observation $n$ and kernel $k \in \{1,\ldots,K\}$ (cf.\ Equation~\ref{eq:cdfmodel}), we draw $S$ samples; for observation $n$ and kernel $k$, sample number $s$ is generated through the inverse cdf of the sigmoid kernel:
\begin{equation}
    {\tilde{e}}^{n,k,s} \coloneqq x_k + h_k \log\mleft( \frac{u^{n,k,s}}{1 - u^{n,k,s}} \mright),~\;~u^{n,k,s} \sim \mathcal{U}(0,1)
    .
    \label{Eq. sampling}
\end{equation}
With these samples, the choice probability can be approximated by:
\begin{align}
        &\!\!\!\!\!\!\!\!
        \int \hat{\rho}_K \mleft(e_{i_n j_n}\mright) \prod_{j' \in C_n, j' \neq j_n} \hat{F}_K \mleft(V_{i_n j_n} + e_{i_n j_n} - V_{i_n j'}\mright) de_{i_n j_n} \nonumber \\
        =&
        \int
        \sum_{k=1}^K \frac{w_k}{h_k} \sigma\mleft( \frac{e_{i_n j_n} - x_k}{h_k} \mright) \mleft( 1 - \sigma\mleft( \frac{e_{i_n j_n} - x_k}{h_k} \mright) \mright) \nonumber \\
        & \quad \cdot \prod_{j' \in C_n, j' \neq j_n} \hat{F}_K \mleft(V_{i_n j_n} + e_{i_n j_n} - V_{i_n j'}\mright) de_{i_n j_n}  
 \label{Eq. Kernel-level sampled nll} \\
        \approx
        &
        \frac{1}{S} \sum_{s=1}^{S} 
        \underbrace{
            \sum_{k=1}^K w_k
            \prod_{j' \in C_n, j' \neq j_n} \hat{F}_K \mleft(V_{i_n j_n} + {\tilde{e}}^{n, k, s} - V_{i_n j'}\mright)}_{\eqqcolon \hat{P}_{ns}} \eqqcolon \hat{P}_n.  
        \nonumber 
\end{align}
Note that due to Jensen's inequality, taking the logarithm in Equation~\ref{Eq. Kernel-level sampled nll} introduces bias which we correct for through third-order Taylor expansion as in \cite{durbin1997monte}.
This leads to our final approximated negative log-likelihood function:
\begin{equation} \label{Eq. Final approximated nll}
\begin{split}
    \hat{\mathcal{L}} 
    &\coloneqq 
    - \sum_{n}^N \mathrm{log}\mleft(\hat{P}_n\mright) 
    + \frac{\sum_{s=1}^{S} (\hat{P}_{ns} - \hat{P}_n)^2}{2S(S-1) \hat{P}_n^2} 
    - \frac{\sum_{s=1}^{S} (\hat{P}_{ns} - \hat{P}_n)^3}{3S(S-1)(S-2)\hat{P}_n^3}
\end{split}
\end{equation}
Finally, we can minimize the approximated loss $\hat{\mathcal{L}}$ through gradient descent to obtain an estimate for $F_\epsilon$.
Thereby, LCM4Rec can potentially recover any underlying choice model with IID error terms from observation data, without further a priori assumptions.

\subsection{Computational complexity}
We apply a straightforward computation of the loss in equation~\ref{Eq. Final approximated nll} 
which has a computational complexity of $\mathcal{O}\left( K^2NS|C| \right)$. The quadratic scaling in $K$ results from our sampling strategy; as described above, we must sample from the pdf $\hat{\rho}_K$ in a way that is differentiable with respect to the weights $w_k$. As a solution, we conveniently generate samples per kernel, resulting in the sums over kernels and samples-per-kernel in equation \ref{Eq. Final approximated nll}.
In practice, we see no benefit in increasing the resolution beyond $K=10$, because very fine grained peaks in the error function---if existent---barely affect choice probabilities.
Thus, since K remains small, the quadratic complexity is not a significant problem in practice.
Alternatively, linear scaling in $K$ can be achieved by sampling uniformly over the support $\hat{\rho}_K$ as in vanilla Monte-Carlo.
We initially implemented this solution but found it less efficient for small values of $K$ than our straightforward computation.

\subsection{Identification and regularization}

Finally, we tackle a challenge in the prevention of overfitting our method:
The variability of $l$ and $c_j$ renders L2-regularization on the user- and item-embeddings $\hat{u}_{i}$ and $\hat{v}_{j}$ ineffective, because for any decrease in $\hat{u}_{i}$ or $\hat{v}_{j}$, the unregularized state can be reconstructed by shrinking $l$, $h_k$ and $c_j$ accordingly. 
Consequentially, L2-regularization would lead to a collapse of $\hat{\rho}_K$.
Thus, in order to prevent overfitting, we scale the item-specific constants $c_j$ into the interval $[0,1]$ after every gradient update.
This forces the width parameter $l$ and the utilities to stay within the same scale as $c_j$.\footnote{Fixing the scale parameter $l$ does not suffice as the model could still make the estimated distribution infinitely thin by placing the weight on only a few design points.}

\section{Experimental Setup}\label{Sect. Experiments}
We perform two experiments to answer three research questions.
A key property of LCM4Rec is that it does not assume what error distribution underlies user choices, thus its accuracy should not depend on them.
Accordingly, our first research question is:
\begin{enumerate}[label=\bfseries RQ\arabic*, wide, labelwidth=!, labelindent=0pt]
    \item \textit{Does LCM4Rec return accurate and robust choice probabilities under different true choice models, i.e., different error distributions? }\label{rq:accuraterobust}
\end{enumerate}
However, accurate choice probabilities do not necessarily indicate the correct choice model is identified, since an incorrect model can still overfit, and thereby, (partially) compensate for an incorrect error distribution.
To address this possibility, we also directly validate whether LCM4Rec finds the correct choice model:
\begin{enumerate}[resume,label=\bfseries RQ\arabic*, wide, labelwidth=!, labelindent=0pt]
    \item \textit{Can LCM4Rec accurately recover the true choice model from interaction data?} \label{rq:recover}
\end{enumerate}
Accurate choice models are more robust to exposure bias, especially from item-co-exposure~\citep{krause2024mitigating}, due to better modeling competition between items.
Thus, our third research question concerns:
\begin{enumerate}[resume,label=\bfseries RQ\arabic*, wide, labelwidth=!, labelindent=0pt]
\item \textit{Is LCM4Rec more robust to exposure bias than parametric choice models?} \label{rq:exposurebias}
\end{enumerate}
To pursue our research questions and understand the importance of assumptions about user behavior, our experiments compare our LCM4Rec to parametric alternatives, in cases where the parametric models and the correct error distribution match, and in cases where they do not match.
Our data and implementation are publicly available at \url{https://github.com/krauthor/LCM4Rec_RecSys2025}.

\subsection{Experiment 1: Accuracy and robustness}
\label{Subsect. Experiment performance}
To assess \ref{rq:accuraterobust} and \ref{rq:recover}, we apply a setup similar to that of \citet{krause2024mitigating}, with the important difference that we vary the choice model underlying interactions.

\parahead{Dataset}
Our first generated dataset consists of the choices of 500 users on 500 items, where every user performed up to 500 choices from uniformly random subsets of four items.
Users could interact with the same item multiple times.
We chose the large number of interactions per users to investigate systematic, asymptotic errors that result from learning with inaccurate choice models.
Each user and item are represented by an embedding, $u_i, v_j \in \mathbb{R}^3$ respectively, sampled uniformly from the unit sphere of radius $\sqrt{2}$.
Item-specific constants $c_j \in \mathbb{R}$ are sampled from the uniform distribution: $\mathcal{U}(0,1)$. 
Final utilities are $U_{ij} \coloneqq u_i \cdot v_j + c_j + \epsilon_{ij}$ (cf.\ equation \ref{eq:utilitymodel}).

We use three different error distributions, and thereby different choice models, to generate the $\epsilon$ error terms:
\begin{enumerate*}[label=(\roman*)]
    \item \textbf{Gumbel} with $\epsilon \sim \mathrm{Gumbel}(0, 0.75)$;
    \item \textbf{Signed exponential} with $\epsilon \sim -\mathrm{Exponen\text{-}}$ $\mathrm{tial}(0, 0.75)$;
    \item \textbf{Gaussian Mixture} with $\epsilon = \frac{1}{3} N_1 + \frac{2}{3} N_2$, where $N_1 \sim \mathcal{N}(-0.75, 0.25)$ and $N_2 \sim \mathcal{N}(0.75, 0.25)$.
\end{enumerate*}
The Gumbel distribution was chosen to represent non-conservative choice behavior, whereas the signed exponential distribution represents conservative choice behavior. 
Lastly, the Gaussian mixture distribution is included because it lacks a corresponding parametric model.
Thus, its output can only be fit by non-parametric models such as ours.

To properly evaluate choice models, it is crucial that the training set contains no user-item interactions that are later used for evaluation, including negative interactions.
We follow \citet{krause2024mitigating} and construct the training and validation datasets based on the interactions of a subset of users $U_\mathrm{Train} \subset U$ on all items together with the interactions of the remaining users $U_\mathrm{Eval} = U \setminus U_\mathrm{Train}$ with half of all items $J_\mathrm{Train} \subset J$. The test set was constructed based on the interactions of $U_\mathrm{Eval}$ with the remaining half of all items $J_\mathrm{Eval} = J \setminus J_\mathrm{Train}$.

\parahead{Models}
The following models are included in our comparison:
\begin{itemize}[leftmargin=*]
    \item The \textbf{multinomial logit model (MNL) \cite{train2009discrete}} 
     is the de-facto standard discrete choice model~\cite{train2009discrete}.
    It assumes \textbf{Gumbel}-distri\-buted errors. 
    Its choice probabilities equal a softmax over the deterministic utilities
    with non-conservative cannibalisation patterns. 
    \item The \textbf{exponomial model (ENL) \cite{alptekinouglu2016exponomial}} 
      assumes that the random errors are generated by a \textbf{signed exponential} distribution with very conservative cannibalization patterns. 
    \item The \textbf{binary logit model (BL) \cite{train2009discrete}} 
    is the univariate variant of the MNL model. 
    It minimizes a logistic loss over all observed positive and negative interactions individually.
    Thus, it ignores any possible competitive effects of alternative options.
    \item The \textbf{binary cross-entropy (BCE) with negative sampling} 
    is a basic matrix factorization recommendation model that is equivalent to the BL model except that it samples the negative labels from the item corpus. 
    \item The \textbf{generalized binary cross-entropy loss (gBCE) \cite{petrov2023gsasrec}} 
    is designed to capture confidence scores more accurately than the BCE loss and the MNL with negative sampling. 
    \item Our \textbf{learned choice model for recommendation (LCM4Rec)} features non-parametric kernel density estimation to infer the correct choice probabilities and cannibalization pattern from the data,
    uniquely, without assuming an error distribution a priori.
\end{itemize}

\parahead{Hyper-parameters}
We set the embedding size for each model to match the size used in the generation process.
Learning rates and optimizer choices are determined via grid-search on the Gumbel distributed data and applied to all distributions.
For LCM4Rec we set $K=5$, $S=5$, and for gBCE we set $t=1$. 
We tried many values of K without running into any computational constraints, optimization instability, or notable overfitting. Values larger than $K=5$ did not result in better predictive accuracy, thus we chose this value to show that a small number of kernels suffice and improvements can be achieved at low computational costs.

\parahead{Evaluation}
To measure the accuracy of the choice probabilities, we use the Kullback-Leibler divergence (KLD) on the distributions over the entire evaluation item set $J_\mathrm{Eval}$.
These distributions represent the users' preferences if they could choose from the entire evaluation corpus and would be relevant during inference \citep{g2024overcoming}.
We measure predictive performance in the forms of the nDCG, negative log-likelihood (NLL), and accuracy (Acc) for the choices among all alternatives in the respective choice sets on the test set. 
For the univariate models we compute the probability of choosing an item from a choice set as the probability of choosing only that item and rejecting the others.
All results are means over 100 simulation repetitions. 
Bootstrapped one-sided t-tests determine statistical significance of performance differences to the next worse value.

\subsection{Experiment 2: Correcting for exposure bias}
\label{Subsect. Experiment bias}
To analyze how the models react to exposure bias for \ref{rq:exposurebias}, we evaluate their learned behavior under different exposure distributions.

\parahead{Dataset}
Our approach to measuring the effect of exposure is to take a small item subset $J^{\mathrm{Bias}} \subset J$, vary its exposure, and subsequently see whether the items are treated differently as a result.
Again, we prevent information leakage by avoiding interactions in the training set with user-item pairs that are also used during evaluation (see Section~\ref{Subsect. Experiment performance}).
The three error distributions are reused: Gumbel, signed exponential and Gaussian Mixture, for each we generate two pairs of two datasets with different exposure distributions: $O'$ and $O''$.
Subsequently, by comparing the pairs of learned models on the test set, 
we can see whether the change in exposure results in a difference in the learned choice models.

Our pairs of exposure distributions follow those proposed by \citet{krause2024mitigating}, to capture different kinds of exposure bias.
The first pair considers the effect of \emph{non-uniform exposure frequencies}, i.e., some items getting more exposure than others.
To generate $O'$, all items are presented equally often; whereas for $O''$, we first randomly select 25 items and then include at least one of them in every second choice set, resulting in 1.9 times as much exposure per item.
The second pair considers bias due to \emph{non-uniform competition}, i.e., the effect of being presented together with more or less popular alternatives.
For this pair, 25 items are randomly selected again, to generate $O'$ these are only presented together with the top 20 percent most-popular items, for $O''$ they are only presented with the bottom 20 percent most popular items (thus the least-popular items).
In both cases, if the choice model is robust to exposure bias, the learned preferences for the selection of 25 items should be  the same for $O'$ and $O''$ in expectation.
Therefore, large differences indicate a substantial effect from exposure, and thus, susceptibility to exposure bias, and vice-versa, small differences indicate robustness.

The number of users, items, and choices per user, the choice set size, the embeddings, and the error distributions are the same as for the first experiment (Section~\ref{Subsect. Experiment performance}). The training, validation and test set were also constructed in the same way.

\parahead{Models in comparison}
For the second experiment, we include the MNL, ENL, and LCM4Rec models.
The univariate models are excluded because previous work has already shown them to be more vulnerable to exposure bias than MNL and ENL~\cite{krause2024mitigating}.

\parahead{Hyper-parameters}
We employ the same hyper-parameters as the previous experiment except that we apply the SGD optimizer to ENL.
Preliminary experiments showed that, for ENL, the Adam optimizer is more sensitive to exposure bias from non-uniform competition without significantly improving accuracy. 

\parahead{Evaluation}
Our evaluation procedure also follows \citet{krause2024mitigating}.
To measure the effect of the simulated exposure bias we consider the average rank of the selection of items as predicted by a model trained on $O'$ or $O''$.
Our metric is the difference in predicted ranks for each pair of equivalent models where one was trained on $O'$ and the other on $O''$.
Thereby, this difference can show whether the exposure bias, i.e., non-uniform exposure or non-uniform competition, can affect the preferences inferred by the choice models.
All reported results are averages over 100 simulation repetitions. 

\subsection{Assumptions and possible extensions}
While we keep modeling assumptions general, several specifications common in real world data can be made:
First, we generated 500 choices per user with 4 options per choice, which is realistic in only some scenarios. 
Non-parametric estimators typically converge slower than parametric ones so that an evaluation on more sparse data sets would be of interest.
Second, we sampled choice sets uniformly. Real-world data often exhibits a long-tail item distribution and how such a distribution could affect our model's robustness remains open. 
Likewise, we sampled embeddings uniformly while a long-tailed utility distribution could be more realistic. Third, we assumed knowledge over users' choice sets. The sets' compositions could be partially unobserved in real-world applications as fully tracking observed options can be complex and as users can discover and interact with items outside of the observable domain. To account for such behavior one can introduce a \emph{no-choice alternative} that represents the decision \emph{not} to choose any presented option \citep{eide2021finn, train2009discrete}. Moreover, padding and cropping allows processing choice sets of varying sizes.

\section{Results and Discussion}\label{Sect. Discussion}

\begin{table}[t]
  \caption{Inference performance and robustness. The correctly specified models are highlighted with $^\dagger$ and were expected to perform best. The KLD scores refer to the choice probabilities from the entire item corpus $J_B$. Best results in bold font, second best underlined. Significantly better scores than the next worse models' follow $\mathbf{^{*}p<0.1}$; $\mathbf{^{**}p<0.05}$; $\mathbf{^{***}p<0.01}$.}
  \label{Tab. Performance on synthesized data}
  {\renewcommand*{\arraystretch}{1.2}
  \begin{tabular}{llllll}
\toprule
 &  & KLD ↓ & NLL ↓ & nDCG ↑ & Acc ↑ \\
\midrule
\multirow[t]{6}{*}{{\rotatebox{90}{\hspace{-1.3cm}Gumbel}}}
& BL & \textbf{0.028}$^{***}$ & \textbf{1.184}$^{***}$ & \textbf{0.996}$^{***}$ & \textbf{0.477}$^{***}$ \\
& BCE & 0.122$^{***}$ & 1.242$^{***}$ & \underline{0.983}$^{***}$ & 0.451$^{***}$ \\
 & gBCE & 0.163$^{***}$ & 1.272 & 0.967 & 0.415 \\
\cline{2-6}
 & MNL$^\dagger$ & \textbf{0.028}$^{***}$ & \textbf{1.184}$^{***}$ & \textbf{0.996}$^{***}$ & \textbf{0.477}$^{***}$ \\
 & ENL & 7.952 & \underline{1.188}$^{***}$ & \textbf{0.996}$^{***}$ & \underline{0.475}$^{***}$ \\
 & LCM4Rec & \underline{0.085}$^{***}$ & \textbf{1.184}$^{***}$ & \textbf{0.996}$^{***}$ & \textbf{0.477}$^{***}$ \\
\cline{1-6}
\multirow[t]{6}{*}{{\rotatebox[origin=br]{90}{\parbox{1.4cm}{\centering Signed\\ Exponential}}}}
 & BL & 1.574$^{***}$ & 1.063 & \textbf{0.997}$^{*}$ & \underline{0.538}$^{***}$ \\
& BCE & 2.310$^{***}$ & 1.146$^{***}$ & \underline{0.986}$^{***}$ & 0.510$^{***}$ \\
 & gBCE & 2.386 & 1.183 & 0.974 & 0.474 \\
\cline{2-6}
 & MNL & {1.491}$^{***}$ & \underline{1.062}$^{*}$ & \textbf{0.997}$^{***}$ & \textbf{0.539}$^{*}$ \\
 & ENL$^\dagger$ & \textbf{0.345}$^{***}$ & \textbf{1.056}$^{***}$ & \textbf{0.997}$^{***}$ & \textbf{0.539}$^{**}$ \\
 & LCM4Rec & \underline{0.459}$^{***}$ & \textbf{1.056}$^{***}$ & \textbf{0.997}$^{***}$ & \textbf{0.539}$^{*}$ \\
\cline{1-6}
\multirow[t]{6}{*}{{\rotatebox[origin=br]{90}{\parbox{1.3cm}{\centering Gaussian\\ Mixture}}}}
 & BL & 1.166$^{***}$ & 1.127$^{***}$ & \textbf{0.997}$^{***}$ & \underline{0.515}$^{***}$ \\
 & BCE & 1.815$^{***}$ & 1.199$^{***}$ & \underline{0.985}$^{***}$ & 0.482$^{***}$ \\
 & gBCE & 1.881 & 1.231 & 0.971 & 0.446 \\
\cline{2-6}
 & MNL & 1.108$^{***}$ & \underline{1.126} & \textbf{0.997}$^{***}$ & \underline{0.515}$^{***}$ \\
 & ENL & \underline{0.625}$^{***}$ & 1.130$^{***}$ & \textbf{0.997}$^{***}$ & 0.514$^{***}$ \\
 & LCM4Rec & \textbf{0.337}$^{***}$ & \textbf{1.123}$^{***}$ & \textbf{0.997}$^{***}$ & \textbf{0.516} \\
\bottomrule
\end{tabular}

  }
\end{table}

\begin{table}[t]
  \caption{Mean KLDs and standard deviation between the true and modeled error distributions.
  $^\dagger$Correctly specified model.
  }
  \label{Tab. KLDs distributions}
  \centering
  \begin{tabular}{llll}
\toprule
 & Gumbel & Sign. Exp. %
 & Gauss. Mix. \\
\midrule
MNL & 0.00$^\dagger$ & 1.25 & 0.43 \\
ENL & 1.19 & 0.00$^\dagger$ & 0.87 \\
LCM4Rec (ours) & {0.13 $\pm$ 0.07} & {0.26 $\pm$ 0.05} & {0.31 $\pm$ 0.18} \\
\bottomrule
\end{tabular}

\end{table}

\begin{figure*}[t]
    \centering
    \includegraphics[width=\textwidth]{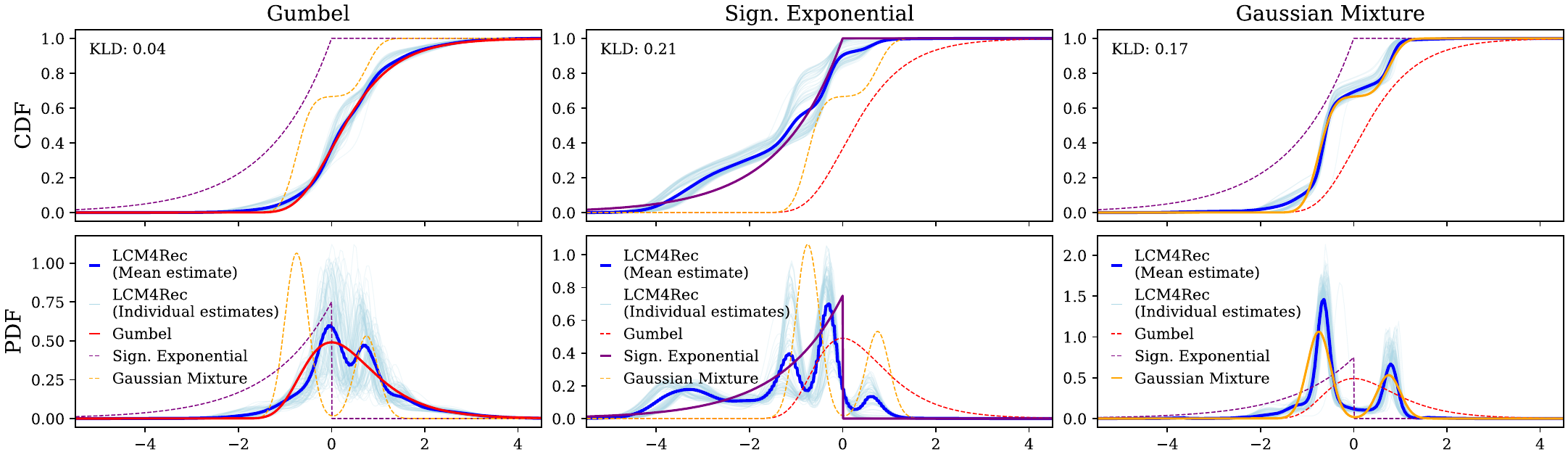}
    \caption{Estimated versus true cdfs and pdfs and KLD score between the mean estimate and the true distribution, smoothed to avoid null sets.
    Estimates are shifted to minimize the KLD to the true distribution (translation invariance in choice models).
    }
    \label{Fig. Simulation estimated pdfs}
\end{figure*}

\subsection{Accuracy and robustness} \label{Subsect. Results RQ1}

Our discussion starts by considering \ref{rq:accuraterobust}: whether LCM4Rec returns accurate and robust choice probabilities under different true choice models.
Table \ref{Tab. Performance on synthesized data} displays the performance of all models in the comparison under three different true choice models.

We see that two of the univariate models BCE and gBCE consistently achieve the worst performance across all error distributions and metrics (except ENL has worse KLD under Gumbel).
This is unsurprising, since univariate models are not designed for the choice model setting and ignore the effect of competition.

Surprisingly, the simplest univariate model, BL, consistently outperforms BCE and gBCE.
In particular, under the Gumbel distribution, BL reaches the highest performance and even outperforms the multivariate ENL on all metrics except nDCG.
BL is the univariate variant of MNL which assumes a Gumbel distribution, which may explain why BL performs so well in this setting.
Under the other distributions, BL no longer reaches the highest performance for KLD, NLL and Acc.
It seems that, whilst overall the multivariate models are more accurate in our setting, univariate models can still benefit from matching the true error distribution.

Next we consider the parametric choice models: MNL and ENL.
Clearly, we observe that each model has the best performance across all metrics when their assumed error distribution matches the true distribution, i.e., Gumbel for MNL and signed exponential for ENL.
This is expected, since the structure of the fitted model matches the true model structure.
Surprisingly, MNL and ENL always reach the highest observed nDCG score and are close the highest accuracy, regardless of the true error distribution.
The NLL is more affected but differences are still quite marginal: always below $0.08$ of the best NLL.
However, for KLD this is not the case: MNL and ENL have substantially worse KLD when their assumed error distribution is wrong.
It thus appears that the MNL and ENL are only able to compensate for their incorrect choice model assumption for the NLL, nDCG, and accuracy metrics.
But this results in substantial degraded match between the learned and true preferences in terms of KLD.
Importantly, we note that our dataset generation is idealized as choice sets are uniform random samples, potentially, this compensation does not occur in less ideal settings.

Finally, we discuss our non-parametric multivariate LCM4Rec model which is the only model that estimates the error distribution.
Strikingly, LCM4Rec reaches the best observed performance for the NLL, nDCG and Acc metrics across all error distributions.
Moreover, it has the best KLD under the Gaussian mixture distribution and second best for the signed exponential and Gumbel distributions.
Therefore, it appears LCM4Rec is only outperformed by parametric models when they correctly assume the true choice model, where the difference in KLD is still limited.
In the case of the Gaussian mixture, where there is no matching choice model, LCM4Rec has significantly better performance in terms of KLD and NLL than all other models.
As a result, LCM4Rec is the only model that has consistent good performance across all distributions and metrics.
We attribute this to the fact that LCM4Rec infers the most likely error distribution, avoiding assumptions about the true choice model.

Therefore, we can answer \ref{rq:accuraterobust} as follows: LCM4Rec is the only model that can accurately predict choice probabilities regardless of the true error distribution.
Thereby it is robust to the true user choice model and the safest choice when the exact model of user behavior is unknown.
Additionally, our results indicate that the robustness of LCM4Rec does not come with a trade-off in predictive performance, as it reaches competitive performance in all metrics. 

Lastly, we note that all our discrete choice models scored comparably well in terms of nDCG.
Thus, it appears the underlying choice model does not seem to critically affect ranking performance in our first experiment setting. 
Moreover, performance of BL shows that, under uniform exposure, exact information on item co-exposure and choice sets is not even needed for ranking purposes.
We believe that this is an artifact of our setup where item exposure is uniform (in expectation), hence our second experiment concerns a setting where this is not the case (see Section \ref{Subsect. Results RQ3}).

\begin{table}
  \caption{Exposure bias. Values show how many ranks items are ranked higher due to overexposure/competition, including 95\%-CIs.
  Best (lowest abs.) bold, second best underlined.}
  \label{Tab. Exposure bias on synthesized data}
  {\renewcommand*{\arraystretch}{1.3}
  
\begin{tabular}{llcc}
\toprule
$\epsilon$ & Model & Overexposure & Competition \\
\midrule
\multirow[t]{3}{*}{{\rotatebox{90}{\hspace{-0.85cm} \small  Gumbel}}} & MNL & -0.427 ± .401 & -1.314 ± .427 \\
 & ENL & \underline{-0.418 ± .461} & 1.701 ± .486 \\
 & LCM4Rec (ours) & \textbf{-0.258 ± .425} & \textbf{-0.000 ± .448} \\
\cline{1-4}
\multirow[t]{3}{*}{{\rotatebox{90}{\hspace{-1.0cm} \small Sign.\ Exp.}}} & MNL & -0.784 ± .311 & -6.423 ± .329 \\
 & ENL & \underline{-0.581 ± .321} & \underline{-2.021 ± .337} \\
 & LCM4Rec (ours) & \textbf{-0.212 ± .351} & \textbf{-1.571 ± .359} \\
\cline{1-4}
\multirow[t]{3}{*}{{\rotatebox{90}{\hspace{-1.15cm} \small Gauss. Mix.}}} & MNL & -0.247 ± .345 & -2.856 ± .381 \\
 & ENL & \underline{-0.100 ± .378} & \underline{0.294 ± .407} \\
 & LCM4Rec (ours) & \textbf{-0.026 ± .363} & \textbf{0.118 ± .389} \\
\bottomrule
\end{tabular}

  }
\end{table}

\begin{figure*}[tp]
    \centering    
    \includegraphics[width=\textwidth]{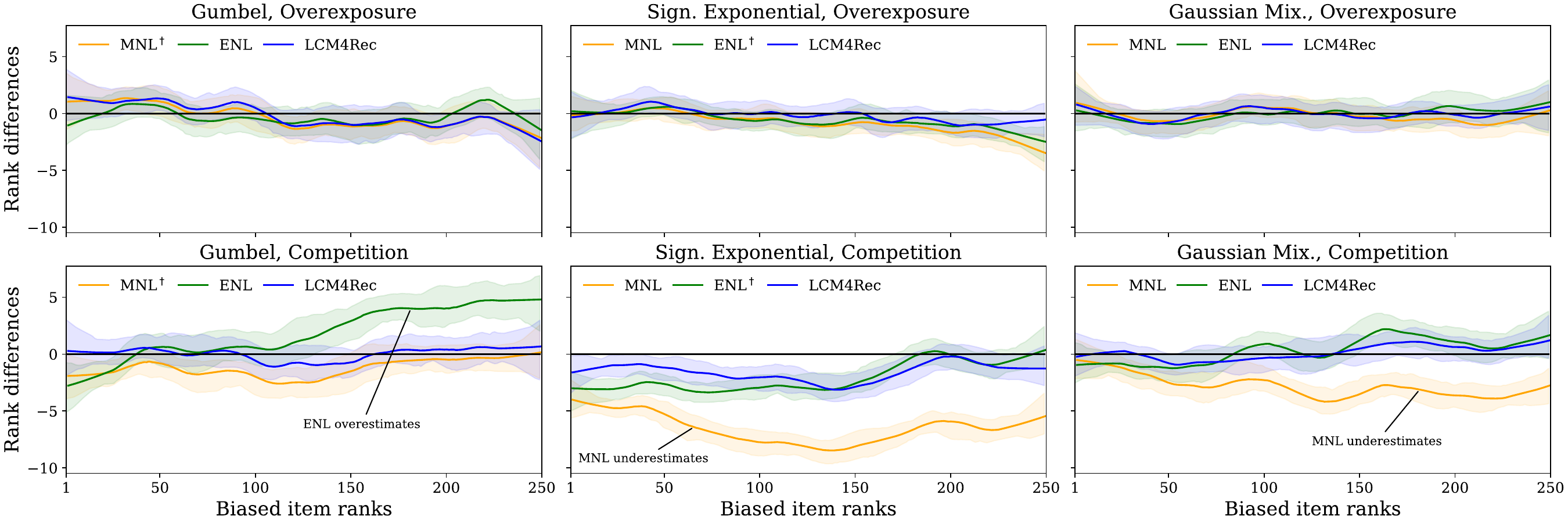}
    \caption{Exposure bias w.r.t. items' mean true ranks. 
    Regression curves and CIs based on LOWESS \citep{cleveland1979robust} and bootstrapping.
    }
    \label{Fig. Exposure bias wrt. item ranks.}
\end{figure*}

\subsection{Learning the choice model} \label{Subsect. Results RQ2}

For \ref{rq:recover}, we evaluate whether LCM4Rec can accurately recover the true error distribution.
Figure~\ref{Fig. Simulation estimated pdfs} displays the error distributions' cdfs and pdfs in our experimental setting and LCM4Rec's estimates.

Upon visual inspection, the estimates accurately approximate the error distributions.
The estimator, which learns the cdf, consistently matches the strongest cdf. 
Its derivatives also follow the true pdfs' shapes.
We observe some oscillation, especially when the error is signed exponential, a likely artifact from the resolution ($K=5$) of our kernel density estimator.
Nevertheless, LCM4Rec correctly identifies Gumbel's right skew, signed exponential's left skew, and the Gaussian mixture's two columns.
The results in Table~\ref{Tab. Performance on synthesized data} also reveal that the estimates produce competitive predictive performance (see Section~\ref{Subsect. Results RQ1}).
To quantify the match, Table \ref{Tab. KLDs distributions} shows the mean KLD scores between the true, the assumed, and the estimated distributions.
LCM4Rec scores low on KLD and much lower than MNL and ENL when their assumed distribution is incorrect.
Hence, we affirm \ref{rq:recover}: LCM4Rec can accurately recover the true error distribution's shape and learn the true user choice model.

\subsection{Exposure bias resistance} \label{Subsect. Results RQ3}
Finally, we turn to \ref{rq:exposurebias} and compare the robustness to exposure bias of LCM4Rec with that of parametric choice models.
Table~\ref{Tab. Exposure bias on synthesized data} shows how many additional ranks items obtain on average in the model predicted ranking of items when they are (i) overexposed versus uniformly exposed; or (ii) presented with popular alternatives versus unpopular alternatives.
Additionally, Figure~\ref{Fig. Exposure bias wrt. item ranks.} displays the difference in observed and true item ranks.

Our results indicate that the effect of  (i) overexposure on the behavior of all three models is limited:
Figure~\ref{Fig. Exposure bias wrt. item ranks.} shows differences are close to zero for all true item ranks.
This is also visible in Table~\ref{Tab. Exposure bias on synthesized data} which also reveals that when considering the confidence intervals of the results, the differences between the models do not appear meaningful.
Surprisingly, LCM4Rec is less affected by bias than both other models, even when their assumed error distribution is correct.
Nevertheless, due to the variance in these results, we conclude that all models appear similarly robust to overexposure.

The effect of competition results in much larger differences in Table~\ref{Tab. Exposure bias on synthesized data}.
The ENL model under Gumbel and the MNL model under signed exponential and Gaussian mixture are heavily affected by bias.
Figure~\ref{Fig. Exposure bias wrt. item ranks.} reveals that more preferred items are affected the most, the bias of the conservative ENL places these items lower in the ranking, while the non-conservative MNL places them higher.
Again, LCM4Rec is the least affected model across all distributions, surprisingly, even when the parametric models match the true distribution.
We speculate that inferring the error distribution and user preferences simultaneously leads to more robustness.

To conclude, we answer \ref{rq:exposurebias} as follows:
Whilst all three choice models appear robust to exposure bias from overexposure, the parametric models can be affected exposure bias from competition depending on what the users' choice model is.
In contrast, our results indicate LCM4Rec is the only model that is robust to both types of bias regardless of what the true users' choice model is.

\section{Conclusion}\label{Sect. Conclusion}

Our work concerns the optimization of multivariate choice models that predict how users choose from a set of recommended items.
Specifically, we addressed the limitation that existing parametric choice models require a priori assumptions about the error distribution that characterizes how users make choices; existing parametric choice models learn user preferences but not \emph{how} choices are made.
Moreover, our experimental results reveal that when their assumptions do not match the true user behavior, they return inaccurate choice probabilities and are susceptible to exposure bias.

In response, we propose LCM4Rec, the first non-parametric multivariate model for recommendation.
In contrast with parametric models, LCM4Rec simultaneously infers the most likely error distribution \emph{and} user preferences.
Thereby, it both learns \emph{what} users prefer and \emph{how} they choose.
It optimizes a kernel density estimator that can approximate any distribution to find the most likely error distribution underlying the users' choice model.
As a result, LCM4Rec alleviates the need to make assumptions about the correct users' choice model.
Our experimental results show that LCM4Rec successfully recovers the correct choice model from observed user interactions.
Furthermore, our results indicate that---~across all tested users' choice models~---only LCM4Rec has consistent competitive predictive performance and is robust to exposure bias coming from competition.
Therefore, we conclude that LCM4Rec provides the most robust way to learn user preferences from their observed choices and model users' choice behavior, while avoiding a priori assumptions about user behavior.

Our work presents promising directions for future work:
LCM4\-Rec assumes errors are IID, a future extension could tackle the non-IID case.
Finally, our experiments were limited to synthetic data but LCM4Rec has the potential to find the most likely real-world user choice model from real-world choice set data.

\begin{acks}
This work is supported by the Dutch Research Council (NWO) under grant VI.Veni.222.269.
All content represents the opinion of the authors, which is not necessarily shared or endorsed by their respective employers and/or sponsors.
\end{acks}

\balance
\bibliography{sample-base}

\end{document}